\documentclass[12pt]{article}
\usepackage{pifont}

\usepackage{amsfonts}
\usepackage{latexsym}
\usepackage{amsmath}
\usepackage{amssymb}
\usepackage{color}

\newtheorem{theorem}{Theorem}[section]

\newtheorem{proposition}{Proposition}[section]

\newtheorem{example}{Example}[section]

\newenvironment{proof}[1][Proof]{\noindent \textbf{#1.} }{\ \ \  $\Box$}

\newtheorem{lemma}{Lemma}[section]

\newtheorem{remark}{Remark}[section]

\title{BSDEs with random default time and their applications to default risk\thanks{This work is partially supported by the National Basic
Research Program of China (973 Program) (Grant No. 2007CB814900)
(Financial Risk).} }

\date{}

 \author{ Shige Peng, Xiaoming Xu\thanks{Corresponding author, E-mail: xmxu@mail.sdu.edu.cn}
 \\ \small{School of Mathematics, Shandong University, Jinan, 250100, China}
 }

\begin{document}

\maketitle

\begin{abstract}
In this paper we are concerned with backward stochastic differential
equations with random default time and their applications to default
risk. The equations are driven by Brownian motion as well as a
mutually independent martingale appearing in a defaultable setting.
We show that these equations have unique solutions and a comparison
theorem for their solutions. As an application, we get a
saddle-point strategy for the related zero-sum stochastic
differential game problem.\\
\par $\textit{Keywords:}$ Backward stochastic differential equation, Random default time,
Comparison theorem, Zero-sum stochastic differential game
\end{abstract}



\section{Introduction}\label{sec:intro}

Credit risk is a kind of the most fundamental, most ancient and most
dangerous financial risk. Particularly in recent years it has been
greatly concerned once more. The most extensively studied form of
credit risk is the default risk, that is, the risk that a
counterpart in a financial contract will not fulfil a contractual
commitment to meet her/his obligations stated in the contract. Many
people, such as Bielecki, Jarrow, Jeanblanc, Kusuoka and so on, have
worked on this subject (see e.g. [2-4, 10, 15-17]).

In a defaultable market, the noise is created by the Brownian motion
$B$ as well as a random time $\tau$ which is referred to as a
default time. Then the information at time $t$ we can get is of two
kinds: one from the assets prices, generated by $B_t$ and denoted by
$\mathcal{F}_t$, the other from the default time, generated by the
default process $H_t:=1_{\{\tau \leq t\}}$ and denoted by
$\mathcal{H}_t$. It should be noted here that the default time
$\tau$ is not an $\mathcal{F}$-stopping time in general. The
filtration we consider is the so-called enlarged filtration
$\mathcal{G}:=\mathcal{F}\vee\mathcal{H}$. Then how do we deal with
this case? Roughly speaking, we construct a process $\Gamma$, called
the $\mathcal{F}$-hazard process of $\tau$,  by setting
$\Gamma_t:=-\ln [1-P(\tau\leq t)]$ where $P$ is the historical
probability measure. Then the process $M$, defined by
$M_t:=H_t-\Gamma_{t\wedge\tau}$, is a $\mathcal{G}$-martingale
independent of $B$. Assume that $\Gamma$ is absolutely continuous,
then there exists an $\mathcal{F}$-adapted process $\gamma$, called
the intensity process, such that $\Gamma_t=\int_0^t \gamma_s ds$. By
the well-known Kusuoka's martingale representation theorem, which
states that any $\mathcal{G}$-square integrable martingale can be
represented as the sum of integrals w.r.t $B$ and $M$, we know that
in a defaultable setting, $B$ and $M$ are of great importance.

When studying the utility maximization problem in a defaultable
setting, Bielecki et al. \cite{BJR1} and Lim-Quenez \cite{LQ}
conclude that the value function is a solution of a BSDE with a
quadratic driver, which we call BSDE with random default time in
this paper. Actually this type of BSDEs appears very naturally. For
the evaluation/hedging problem, Bielecki et al. studied the PDE
approach in \cite{BJR3} (see also \cite{BJR5}), where it is assumed
that the defaultable market is complete and the dynamics of the
primary assets are linear SDEs driven by both $B$ and $M$. Their
goal is to replicate a contingent claim $\xi$ which depends on
whether the default event occurs or not. In fact, we know already
that the theory of contingent claim valuation in a complete
default-free market (see e.g. Black-Scholes \cite{BS}, Merton
\cite{M} and so on) can be expressed in terms of classical BSDEs.
Here we will check detailedly in the text that the
evaluation/hedging problem of $\xi$ can be represented as a linear
BSDE with random default time $\tau$ of the following form:
\begin{equation*}\label{equation:intro linear}
Y_t=\xi+\int_t^T (u_s Y_s+ v_s Z_s+ w_s1_{\{\tau
> s\}}\gamma_s \zeta_s)ds-\int_t^T Z_s dB_s-\int_t^T \zeta_s dM_s,
\end{equation*}
which can be solved thanks to the existence of the risk-neutral
measure $Q$ equivalent to the historical probability $P$. In fact
$Q$ is of the following form:
$$Q=\exp[\int_0^T \ln(1+w_s)dH_s-\int_0^T w_s1_{\{\tau > s\}}\gamma_s ds+\int_0^T v_s dB_s-\frac{1}{2}\int_0^T v_s^2ds+\int_0^T u_s ds].$$
Then we have $Y_0=E_Q[\xi]$ which is called the fair price of $\xi$.
While in general, we do not know the exact values of $(u, v, w)$ but
a set $\Theta$ they belong to, which will lead to model uncertainty
or ambiguity (see e.g. [6, 9] for details). Then in this case,
instead of having only one risk-neutral probability measure $Q$
fixed, we will face an uncertain subset of probability measures
$\{P_\theta: \theta=(u, v, w)\in \Theta\}$. For this situation a
robust way to evaluate $\xi$ is its upper price $\hat{Y}_0$ achieved
by a superhedging strategy and $\hat{Y}_0$ can be calculated by
$$\hat{Y}_0=\sup_{\theta\in \Theta}Y_0^\theta=\sup_{\theta\in
\Theta}E_{P_\theta}[\xi],
$$ where $Y_0^\theta$ is in fact the fair price for $\xi$ in a fictitious
market.  In evaluation/hedging problem with this imprecise knowledge
of the risk-neutral measure, we will face a nonlinear BSDE with
random default time (the general form):
\begin{equation}\label{equation:intro}
Y_t=\xi+\int_t^T g(s, Y_s, Z_s,\zeta_s)ds-\int_t^T Z_s dB_s-\int_t^T
\zeta_s dM_s.
\end{equation}
It is worth noting that for the calculation of the upper price
$\hat{Y}_0$, the generator $g$ is given by $$g(s, y, z,
\varsigma)=\sup_{(u_s, v_s, w_s)\in \Theta} (u_s y+ v_s z+ w_s
\varsigma),$$ which can be easily seen from Section \ref{sec:game}.

We are interested in the problem of the existence and uniqueness of
a solution for (\ref{equation:intro}), that is, whether there exists
a unique triple of $\mathcal{G}$-adapted processes $(Y, Z, \zeta)$
satisfying (\ref{equation:intro}).

It is well-known that, in the framework of Brownian filtration, the
general form of BSDE was firstly studied by Pardoux-Peng \cite{PP1}.
Since then, the theory of BSDEs has been studied with great
interest. One of the achievements of this theory is the comparison
theorem. It is due to Peng \cite{P1} and then generalized by
Pardoux-Peng \cite{PP2}, El Karoui et al. \cite{KPQ}. It allows to
compare the solutions of two BSDEs whenever we can compare the
terminal conditions and the generators. These results are applied
widely to default-free markets. For example, BSDE was firstly
applied to the problem of zero-sum stochastic differential games by
Hamadene-Lepeltier \cite{HL1}. From then on, BSDEs were linked with
the game problems closer and closer (see e.g. [8, 12]).

In this paper, we show that under proper assumptions, BSDE
(\ref{equation:intro}) has a unique solution. Besides we also
establish a comparison theorem. It should be noted here that, the
comparison theorem needs one more condition for the generator than
the existence and unique theorem, which is different from the
classical case. As an application, we deal with a zero-sum
stochastic differential game problem, which can also be seen as a
utility maximization problem under model uncertainty. For the game,
we assume that there are two players $J_1$ and $J_2$ whose
advantages are antagonistic. The dynamics of the controlled system
is
\begin{eqnarray*}
X_t &=& x_0+\int_0^t(b(s, X_{s-}, u_s, v_s)+c(s, X_{s-}, u_s,
v_s)1_{\{\tau> s\}}\gamma_s) ds\\
&& +\int_0^t\sigma(s, X_{s-})dB_s+\int_0^t \kappa(s, X_{s-})dM_s.
\end{eqnarray*}
The player $J_1$ (resp. $J_2$) chooses a control $u$ (resp. $v$).
The object of $J_1$ (resp. $J_2$) is  to minimize (resp. maximize)
the cost functional $J^{u, v}$. In this paper, we show that there
exists a saddle point $(u^*, v^*)$ such that $J(u^*, v) \leq J(u^*,
v^*) \leq J(u, v^*)$ for each $(u, v)$.

The paper is organized as follows: in Section \ref{sec:pre}, we list
some notations and assumptions we will use. In Section
\ref{secsec:bsde}, we will first start with a simple model following
\cite{BJR3}, which in fact implies the new idea, that is, BSDE with
random default time, for credit risk modeling. Then we prove an
existence and uniqueness result for BSDEs with random default time
and also establish a comparison theorem. In the last section (i.e,
Section \ref{sec:game}), we solve a zero-sum stochastic differential
game problem in a defaultable setting as an application of the study
of the previous section. For reader's convenience we present some
basic results in the Appendix.

\section{Notations and assumptions}\label{sec:pre}

Let $\{ B_t; t\geq 0 \}$ be a $d$-dimensional standard Brownian
motion on a probability space $(\Omega, \mathcal{F}, P)$ and
$(\mathcal{F}_t)_{ t\geq 0}$ be its natural filtration. Denote by
$|\cdot|$ the norm in $\mathbb{R}^m$.

Let $\{\tau_i; i=1, 2, \ldots, k\}$ be $k$ nonnegative random
variables satisfying
$$P(\tau_i > 0)=1;\ P(\tau_i >t )>0, \forall t
\in \mathbb{R}_+;\ \ P(\tau_i= \tau_j)=0\ (i\neq j).$$ For each i,
we introduce a right-continuous process $\{H_t^i; t\geq 0\}$ by
setting $H_t^i:=1_{\{ \tau_i \leq t \}}$ and denote by
$\mathbb{H}^i=(\mathcal{H}_t^i)_{ t\geq 0}$ the associated
filtration $\mathcal{H}_t^i=\sigma(H_s^i:0 \leq s \leq t)$.

Just as in the general reduced-form approach, for fixed $T>0$, there
are two kinds of information: one from the assets prices, denoted by
$\mathbb{F}=(\mathcal{F}_t)_{0\leq t \leq T}$, and the other from
the default times $\{\tau_i; i=1,2, \ldots, k\}$, denoted by
$\{\mathbb{H}^i; i=1,2, \ldots, k\}$ from the above. The enlarged
filtration considered is denoted by
$\mathbb{G}=(\mathcal{G}_t)_{0\leq t \leq T}$ where
$\mathcal{G}_t=\mathcal{F}_t \vee \mathcal{H}_t^1 \vee
\mathcal{H}_t^2 \vee \ldots \vee \mathcal{H}_t^k$, which indicates
that each $\tau_i$ is a $\mathbb{G}$-stopping time but not
necessarily an $\mathbb{F}$-stopping time in the general case.

Now we assume the following (see \cite{K}):

$(\bf{A})$ \textit{There exist} $\mathbb{F}$\textit{-adapted
processes} $\gamma^i \geq 0\ (i=1, 2, \ldots, k)$ \textit{such that}
$$M_t^i:=H_t^i-\int_0^t 1_{\{\tau_i >s\}}\gamma_s^ids\ \ (i=1, 2, \ldots, k)$$
\textit{are} $\mathbb{G}$\textit{-martingales under} $P$.

$(\bf{H})$  \textit{Every} $\mathbb{F}$\textit{-local martingale is
a} $\mathbb{G}$\textit{-local martingale}.

It should be mentioned that $(\bf{H})$ is a very general and
essential hypothesis in the area of enlarged filtration (see
\cite{MY}).

The following are just for the sake of simplicity:

(i) notations of vectors:
\begin{eqnarray*}
&& H_t:=(H_t^1, H_t^2, \cdots, H_t^k)^\prime,\ M_t:=(M_t^1, M_t^2,
\cdots, M_t^k)^\prime,
\\
&& 1_{\{\tau>t\}}\gamma_t:=(1_{\{\tau_1>t\}}\gamma_t^1,
1_{\{\tau_2>t\}}\gamma_t^2, \cdots,
1_{\{\tau_k>t\}}\gamma_t^k)^\prime,
\end{eqnarray*} where
$(\cdot)^\prime$ is the transpose;

(ii) notations of sets:

$\bullet$ $L^2(\mathcal{G}_T; \mathbb{R}^m)$ := $\{\xi\in
\mathbb{R}^m$ $|$ $\xi$ is a $\mathcal{G}_T$-measurable random
variable such that $E|\xi|^2< + \infty\};$

$\bullet$ $L_{\mathcal{G}}^2(0, T; \mathbb{R}^m)$ := $\{ \varphi:
\Omega\times [0, T]\rightarrow \mathbb{R}^m$ $|$ $\varphi$ is
progressively measurable and $E\int_0^T |\varphi_t|^2dt< +
\infty\};$

$\bullet$ $S_{\mathcal{G}}^2(0, T; \mathbb{R}^m)$ := $\{\varphi:
\Omega\times [0, T]\rightarrow \mathbb{R}^m$ $|$ $\varphi$ is
progressively measurable and $E[\sup_{0 \leq t \leq T}
|\varphi_t|^2]< + \infty\};$

$\bullet$ $L_{\mathcal{G}}^{2, \tau}(0, T; \mathbb{R}^{m\times k})$
:= $\{ \varphi: \Omega\times [0, T]\rightarrow \mathbb{R}^{m\times
k}$ $|$ $\varphi$ is progressively measurable and $E\int_0^T
|\varphi_t|^2 1_{\{\tau
>t\}}\gamma_tdt:=E\int_0^T \sum_{j=1}^m\sum_{i=1}^k |\varphi_{ji,t}|^21_{\{\tau_i
>t\}}\gamma_t^idt < + \infty \}.$

\section{BSDE with random default time}\label{secsec:bsde}

This section discusses BSDEs with random default time of the general
form. We start by analyzing the following example in a defaultable
financial market.

\subsection{An example}\label{sec:in finance}

At the beginning, we assume that the defaultable market is complete
and arbitrage free,
that is to say, any $\mathcal{G}_T$-measurable random variable is a
tradable contingent claim.

In the remaining part of this subsection, following Bielecki et al.
\cite{BJR3}, we will with a markovian set-up. For simplicity, we
assume that here $k=1$, the density $\gamma$ is a constant, the
trading occurs on the interval $[0, T]$, and the dynamics of primary
assets are
$$dY_t^i=Y_{t-}^i(\mu_idt +\nu_idB_t +\kappa_idM_t),\ i=1, 2, 3,$$
where $\mu_i$, $\nu_i$, $\kappa_i\geq -1$ are constants and the
primary assets may be default-free $(\kappa_i=0)$ or defaultable
$(\kappa_i \neq 0)$. Our goal is to replicate a contingent claim of
the form
$$Y_T=1_{\{\tau \leq T\}}g_1(Y_T^1, Y_T^2, Y_T^3)+1_{\{\tau > T\}}g_0(Y_T^1, Y_T^2, Y_T^3)=G(H_T, Y_T^1, Y_T^2, Y_T^3),$$
which settles at time $T$. From the completeness of the market, we
know that $Y_T$ is replicatable.

Let us now consider a small investor whose actions cannot affect
market prices and who can decide at time $t\in [0, T]$ what amount
$\theta_t^i$ of the wealth $Y_t$ to invest in the $\textit{i}$th
asset, $i=1, 2, 3.$ Of course, his decisions can only be based on
the current information $\mathcal{G}_t$, i.e, the processes
$\theta=(\theta^1, \theta^2, \theta^3)^\prime$ and
$\theta^1=Y-\theta^2-\theta^3$ are predictable. Following
Harrison-Pliska \cite{HP}, we say a strategy is self-financing if
the wealth process satisfies the equality
$$Y_t=Y_0+ \int_0^t \theta_s^1
\frac{dY_s^1}{Y_{s-}^1}+\int_0^t \theta_s^2
\frac{dY_{s}^2}{Y_{s-}^2}+\int_0^t \theta_s^3
\frac{dY_{s}^3}{Y_{s-}^3},$$ or, equivalently, if the wealth process
satisfies the linear stochastic differential equation
$$dY_t=\sum_{i=1}^3 \theta_t^i \mu_i dt+\sum_{i=1}^3 \theta_t^i \nu_i dB_t+\sum_{i=1}^3 \theta_t^i \kappa_i dM_t.$$
Noting $\theta^1=Y-\theta^2-\theta^3$, we have to find a strategy
$\theta$ satisfying
\begin{equation}\label{bsdewd in finance}
\left\{
\begin{tabular}{rll}
$dY_t$ &=& $[\mu_1
Y_t+\theta_t^2(\mu_2-\mu_1)+\theta_t^3(\mu_3-\mu_1)]dt$
\\
&& $ +[\nu_1 Y_t+\theta_t^2 (\nu_2-\nu_1)+\theta_t^3
(\nu_3-\nu_1)]dB_t$
\\
&& $+[\kappa_1Y_t+\theta_t^2 (\kappa_2-\kappa_1)+\theta_t^3
(\kappa_3-\kappa_1)]dM_t,$
\\
$Y_T$ &=& $G(H_T, Y_T^1, Y_T^2, Y_T^3).$
\end{tabular}
\right.
\end{equation}
Let
$$Z_t=\nu_1 Y_t+\theta_t^2 (\nu_2-\nu_1)+\theta_t^3
(\nu_3-\nu_1),\ \zeta_t1_{\{\tau> t\}}=\kappa_1Y_t+\theta_t^2
(\kappa_2-\kappa_1)+\theta_t^3 (\kappa_3-\kappa_1).$$ That is,
$\zeta$ is well defined only on $[0, \tau \wedge T]$, in fact, we
have $\zeta_t1_{\{\tau> t\}}dM_t=\zeta_tdM_t$ since $dM_t=0$ on
$[\tau \wedge T, T]$. Then by simple computation we can get
$\theta^i=\theta^i(Y, Z, \zeta)$ $(i=2, 3)$ where $\theta^i(\cdot,
\cdot, \cdot)$ $(i=2, 3)$ are linear functions of the following
form:
\begin{eqnarray*}
\theta^2&:=& \theta^2(Y, Z, \zeta)=a_2 Y+ b_2 Z+ c_2 1_{\{\tau>
t\}}\zeta\\
&:=&
\frac{\kappa_1(\nu_3-\nu_1)-\nu_1(\kappa_3-\kappa_1)}{(\nu_2-\nu_1)(\kappa_3-\kappa_1)-(\kappa_2-\kappa_1)(\nu_3-\nu_1)}Y\\
&&
+\frac{\kappa_3-\kappa_1}{(\nu_2-\nu_1)(\kappa_3-\kappa_1)-(\kappa_2-\kappa_1)(\nu_3-\nu_1)}Z\\
&&
-\frac{\nu_3-\nu_1}{(\nu_2-\nu_1)(\kappa_3-\kappa_1)-(\kappa_2-\kappa_1)(\nu_3-\nu_1)}1_{\{\tau>
t\}}\zeta,
\end{eqnarray*}
\begin{eqnarray*}
\theta^3&:=& \theta^3(Y, Z, \zeta)=a_3 Y+ b_3 Z+ c_3 1_{\{\tau>
t\}}\zeta\\
&:=&
\frac{\kappa_1(\nu_2-\nu_1)-\nu_1(\kappa_2-\kappa_1)}{(\nu_3-\nu_1)(\kappa_2-\kappa_1)-(\kappa_3-\kappa_1)(\nu_2-\nu_1)}Y\\
&&
+\frac{\kappa_2-\kappa_1}{(\nu_3-\nu_1)(\kappa_2-\kappa_1)-(\kappa_3-\kappa_1)(\nu_2-\nu_1)}Z\\
&&
-\frac{\nu_2-\nu_1}{(\nu_3-\nu_1)(\kappa_2-\kappa_1)-(\kappa_3-\kappa_1)(\nu_2-\nu_1)}1_{\{\tau>
t\}}\zeta.
\end{eqnarray*}
Write $a= \mu_1+(\mu_2-\mu_1)a_2+(\mu_3-\mu_1)a_3,$
$b=(\mu_2-\mu_1)b_2+(\mu_3-\mu_1)b_3$,
$c=\frac{1}{\gamma}[(\mu_2-\mu_1)c_2+(\mu_3-\mu_1)c_3]$, then
(\ref{bsdewd in finance}) becomes
\begin{equation*}
\left\{
\begin{tabular}{rll}
$dY_t$ &=& $(a Y_t+ b Z_t+ c 1_{\{\tau>
t\}}\gamma\zeta_t)dt+Z_tdB_t+\zeta_tdM_t,$
\\
$Y_T$ &=& $G(H_T, Y_T^1, Y_T^2, Y_T^3),$
\end{tabular}
\right.
\end{equation*}
which is just a linear backward stochastic differential equation
with random default time $\tau$. Suppose that $c < 1$, which is in
fact a more general condition than that in \cite{BJR3}. Set
$$Q_t:=\exp[\ln(1-c)H_t+c \int_0^t 1_{\{\tau > s\}}\gamma ds-bB_t-\frac{1}{2}b^2 t-a t],$$
i.e.,
\begin{equation*}
\left\{
\begin{tabular}{rll}
$dQ_t$ &=& $-Q_{t-}(a dt+b dB_t+c dM_t),$\\
$Q_0$ &=& $1.$
\end{tabular}
\right.
\end{equation*}
Applying It\^{o}'s formula (see Appendix) on $Q_tY_t$, we have
\begin{eqnarray*}
dQ_tY_t &=& c Q_t 1_{\{\tau >
t\}}\gamma\zeta_tdt+Q_t(Z_t-b Y_t)dB_t+Q_{t-}(\zeta_t-c Y_{t-})dM_t-c Q_{t-} \zeta_tdH_t\\
&=&Q_t(Z_t-b Y_t)dB_t+Q_{t-}(\zeta_t-c Y_{t-}-c \zeta_t)dM_t,
\end{eqnarray*}
which implies $Q_tY_t=E^{\mathcal{G}_t}[Q_TG(H_T, Y_T^1, Y_T^2,
Y_T^3)]$. In the financial market, $Y$ will be called the fair price
of the contingent claim $G(H_T, Y_T^1, Y_T^2, Y_T^3)$.

\subsection{BSDE with random default time}\label{sec:bsde}

The model mainly discussed in this part is:
\begin{equation}\label{3.1}
Y_t=\xi(H_T)+\int_t^T g(s, Y_s, Z_s,\zeta_s)ds-\int_t^T Z_s
dB_s-\int_t^T \zeta_s dM_s.
\end{equation}

In the defaultable financial market, $\xi(H_T)$ represents a
contingent claim needed to be replicated, settled at time $T$,
depending on the event whether the default occurs at time $T$.
Trading occurs on the interval $[0, T]$. $Z$ and $\zeta$ represent
the information of the hedging strategies, for example, in the
linear case (see Subsection \ref{sec:in finance}), we can compute
the hedging strategies by $Z$ and $\zeta$.

The function $g$ is called the generator of (\ref{3.1}). Our object
is to find a triple $(Y_t, Z_t, \zeta_t)\in S_\mathcal{G}^2(0, T;
\mathbb{R}^m) \times L_\mathcal{G}^2(0, T; \mathbb{R}^{m\times d})
\times L_\mathcal{G}^{2, \tau}(0, T; \mathbb{R}^{m\times k})$
satisfying (\ref{3.1}). For this purpose, we first consider a very
simple case: $g(y, z, \varsigma)$ is a real valued process that is
independent of the variable $(y, z, \varsigma)$.

\begin{lemma}\label{lemma:estimate}
For a fixed $\xi(H_T)\in L^2(\mathcal{G}_T; \mathbb{R})$ and
$g_0(\cdot)$ satisfying $$E[(\int_0^T g_0(s) ds)^2]< + \infty,$$
there exists a unique triple of processes $(y_\cdot, z_\cdot,
\varsigma_\cdot)\in L_\mathcal{G}^2(0, T; \mathbb{R}) \times
L_\mathcal{G}^2(0, T; \mathbb{R}^{d}) \times L_\mathcal{G}^{2,
\tau}(0, T; \mathbb{R}^{k})$ satisfying
\begin{equation}\label{euqation:only t}
y_t=\xi(H_T)+\int_t^T g_0(s)ds-\int_t^T z_s dB_s-\int_t^T
\varsigma_s dM_s.
\end{equation}
If $g_0(\cdot)\in L_\mathcal{G}^2(0, T; \mathbb{R})$, then
$(y_\cdot, z_\cdot, \varsigma_\cdot)\in S_\mathcal{G}^2(0, T;
\mathbb{R}) \times L_\mathcal{G}^2(0, T; \mathbb{R}^{d}) \times
L_\mathcal{G}^{2, \tau}(0, T; \mathbb{R}^{k})$. We have the
following basic estimate:
\begin{equation}\label{estimat1}
\begin{tabular}{rll}
&&
$|y_t|^2+E^{\mathcal{G}_t}\int_t^T[\frac{\beta}{2}|y_s|^2+|z_s|^2+\|\varsigma_s\|_\tau^2]e^{\beta(s-t)}ds$
\\
\\
&& $\leq E^{\mathcal{G}_t}|\xi(H_T)|^2
e^{\beta(T-t)}+\frac{2}{\beta}E^{\mathcal{G}_t}\int_t^T|g_0(s)|^2
e^{\beta(s-t)}ds,$
\end{tabular}
\end{equation}
in particular,
\begin{equation}\label{estimat2}
|y_0|^2+E\int_0^T[\frac{\beta}{2}|y_s|^2+|z_s|^2+\|\varsigma_s\|_\tau^2]e^{\beta
s}ds \leq E|\xi(H_T)|^2 e^{\beta
T}+\frac{2}{\beta}E\int_0^T|g_0(s)|^2 e^{\beta s}ds,
\end{equation}
where
$\|\varsigma_s\|_\tau:=|\varsigma_s|1_{\{\tau>s\}}\sqrt{\gamma_s}=(
\sum_{i=1}^k
|\varsigma^{i}|^21_{\{\tau_i>t\}}\gamma_t^i)^{\frac{1}{2}}$ and
$\beta >0$ is an arbitrary constant.
\\
We also have
\begin{equation}\label{estimat3}
E[\sup_{0\leq t\leq T} |y_t|^2]\leq C_T E[|\xi(H_T)|^2+\int_0^T
|g_0(s)|^2ds],
\end{equation}
where the constant $C_T$ depends only on $T$.
\end{lemma}

\begin{proof}
Define $$N_t=E^{\mathcal{G}_t}[\xi+\int_0^T g_0(s)ds].$$ Obviously
$N_t$ is a square integrable $\mathbb{G}$-martingale. Thanks to
Kusuoka's martingale representation theorem (see Appendix), there
exists a unique pair of adapted process $(z_t, \varsigma_t)\in
L_\mathcal{G}^2(0, T; \mathbb{R}^{d}) \times L_\mathcal{G}^{2,
\tau}(0, T; \mathbb{R}^{k})$ such that
$$N_t=N_0+\int_0^t z_s dB_s+\int_0^t \varsigma_s dM_s.$$
Thus
$$N_t=N_T-\int_t^T z_sdB_s-\int_t^T \varsigma_s dM_s.$$
Denote
$$y_t=N_t-\int_0^tg_0(s)ds=N_T-\int_0^tg_0(s)ds-\int_t^T z_sdB_s-\int_t^T \varsigma_s dM_s.$$
Since $N_T=\xi+\int_0^T g_0(s)ds$, immediately we get
(\ref{euqation:only t}).
\par The uniqueness is a simple consequence of the estimate
(\ref{estimat2}). We only need to prove the priori estimates. To
prove (\ref{estimat1}), we first consider the case where $\xi$ and
$g_0(\cdot)$ are both bounded. Since
$y_t=E^{\mathcal{G}_t}[\xi(H_T)+\int_t^T g_0(s)ds]$, thus the
process $y$ is also bounded.

From the equation (\ref{euqation:only t}), we have
$$dy_s=-g_0(s)ds+z_sdB_s+\varsigma_sdM_s.$$
We then apply It\^{o}'s formula to $y_s^2e^{\beta s}$ (see Example
\ref{eg:integration by parts}) for $s\in [t, T]$:
\begin{eqnarray*}
dy_s^2e^{\beta s} &=& e^{\beta s}(\beta
y_s^2-2y_sg_0(s)+|z_s|^2+|\varsigma_s|^21_{\{\tau>s\}}\gamma_s)ds\\
& &+2e^{\beta s}y_s z_sdB_s+e^{\beta s}(2y_s
\varsigma_s+|\varsigma_s|^2)dM_s.
\end{eqnarray*}
Integrating $s$ from $t$ to $T$ and take conditional expectation
with regard to $\mathcal{G}_t$ on both sides, we obtain
\begin{eqnarray*}
& &|y_t|^2+E^{\mathcal{G}_t}\int_t^T[\beta
|y_s|^2+|z_s|^2+|\varsigma_s|^21_{\{\tau>s\}}\gamma_s]e^{\beta(s-t)}ds\\
& &=E^{\mathcal{G}_t}|\xi|^2e^{\beta (T-t)}+E^{\mathcal{G}_t}\int_t^T 2y_sg_0(s)e^{\beta (s-t)}ds\\
& &\leq E^{\mathcal{G}_t}|\xi|^2e^{\beta
(T-t)}+E^{\mathcal{G}_t}\int_t^T
[\frac{\beta}{2}|y_s|^2+\frac{2}{\beta}|g_0(s)|^2]e^{\beta (s-t)}ds.
\end{eqnarray*}
From this it follows (\ref{estimat1}) and (\ref{estimat2}).
\par We now consider the case where $\xi$ and $g_0(\cdot)$ are
possibly unbounded. We set
$$\xi^n:=(\xi\wedge n)\vee(-n),\ \ \ \ g_0(s):=(g_0(s)\wedge n)\vee(-n),$$
and
$$y_t^n:=\xi^n+\int_t^T g_0^n(s)ds-\int_t^T z_s^n dB_s-\int_t^T \varsigma_s^n dM_s.$$
Thanks to the boundedness of $\xi^n,\ \xi^k,\ g_0^n$ and $g_0^k$ for
each positive integer $n$ and $k$, we have
\begin{equation}\label{equation:n estimate}
\begin{tabular}{rll}
& & $|y_t^n|^2+E^{\mathcal{G}_t}\int_t^T[\frac{\beta}{2}
|y_s^n|^2+|z_s^n|^2+|\varsigma_s^n|^21_{\{\tau>s\}}\gamma_s]e^{\beta(s-t)}ds$
\\ \\& & $\leq E^{\mathcal{G}_t}|\xi^n|^2e^{\beta
(T-t)}+\frac{2}{\beta}E^{\mathcal{G}_t}\int_t^T |g_0^n(s)|^2e^{\beta
(s-t)}ds$
\end{tabular}
\end{equation}
and
\begin{equation*}
\begin{tabular}{rll}
& & $E\int_0^T[\frac{\beta}{2}
|y_s^n-y_s^k|^2+|z_s^n-z_s^k|^2+|\varsigma_s^n-\varsigma_s^k|^21_{\{\tau>s\}}\gamma_s]e^{\beta
s}ds$ \\ \\& & $\leq E|\xi^n-\xi^k|^2e^{\beta
T}+\frac{2}{\beta}E\int_0^T |g_0^n(s)-g_0^k(s)|^2e^{\beta s}ds.$
\end{tabular}
\end{equation*}
The second inequality implies that the processes ${y^n}$, ${z^n}$
and ${\varsigma^n}$ are Cauchy sequences in their corresponding
spaces. Thus (\ref{estimat1}) is proved by letting $n$ tend to $+
\infty$ in (\ref{equation:n estimate}).

Easily we can get $y_t\in S_{\mathcal{G}}^2(0, T; \mathbb{R})$ as
(\ref{estimat3}) is a simple consequence of (\ref{estimat2})
together with B-D-G inequality applied to (\ref{euqation:only t}).
\end{proof}

With the above basic estimates, we can now consider the general case
of (\ref{3.1}). We assume that $g(\omega, t, y, z, \varsigma):
\Omega\times[0, T]\times\mathbb{R}^m \times \mathbb{R}^{m\times
d}\times\mathbb{R}^{m\times k}\rightarrow \mathbb{R}^m$ satisfies
the following conditions:

(a) $g(\cdot, 0, 0, 0)\in L_\mathcal{G}^2(0, T; \mathbb{R}^m)$;

(b) the Lipschitz condition: for each $(t, y, z, \varsigma), (t,
\bar{y}, \bar{z}, \bar{\varsigma})\in [0,
T]\times\mathbb{R}^m\times\mathbb{R}^{m\times
d}\times\mathbb{R}^{m\times k}$, there exists a constant $C\geq 0$
such that
$$|g(t, y, z,
\varsigma)-g(t, \bar{y}, \bar{z}, \bar{\varsigma})|\leq
C(|y-\bar{y}|+|z-\bar{z}|+|\varsigma-\bar{\varsigma}|1_{\{\tau
>t\}}\sqrt{\gamma_t}).$$

\begin{theorem}\label{thm:1-main}
Assume that $g$ satisfies (a) and (b), then for any fixed terminal
condition $\xi(H_T)\in L^2(\mathcal{G}_T; \mathbb{R}^m)$, BSDE
(\ref{3.1}) has a unique solution, i.e, there exists a unique triple
of $\mathcal{G}_t-$adapted processes $$(Y_t, Z_t, \zeta_t)\in
S_\mathcal{G}^2(0, T; \mathbb{R}^m) \times L_\mathcal{G}^2(0, T;
\mathbb{R}^{m\times d}) \times L_\mathcal{G}^{2, \tau}(0, T;
\mathbb{R}^{m\times k})$$ satisfying (\ref{3.1}).
\end{theorem}

\begin{proof}
First we introduce a norm in $L_\mathcal{G}^2(0, T; \mathbb{R}^m)
\times L_\mathcal{G}^2(0, T; \mathbb{R}^{m\times d}) \times
L_\mathcal{G}^{2, \tau}(0, T; \mathbb{R}^{m\times k})$:
$$\|(u., v., w. )\|_\beta\equiv\{E\int_0^T(|u_s|^2+|v_s|^2+\|w_s\|_\tau^2)e^{\beta s}ds\}^{\frac{1}{2}}.$$
We set
$$Y_t=\xi(H_T)+\int_t^Tg(s, y_s, z_s, \varsigma_s)ds-\int_t^TZ_sdB_s-\int_t^T\zeta_sdM_s.$$
We define the following mapping $I$ from $L_\mathcal{G}^2(0, T;
\mathbb{R}^m) \times L_\mathcal{G}^2(0, T; \mathbb{R}^{m\times d})
\times L_\mathcal{G}^{2, \tau}(0, T; \mathbb{R}^{m\times k})$ into
itself: $(Y., Z., \zeta.)=I[(y., z., \varsigma.)].$

Next we will prove that $I$ is a strict contraction mapping under
the norm $\|\cdot\|_\beta$. For any two elements $(y, z, \varsigma)$
and $(\bar{y}, \bar{z}, \bar{\varsigma})\in L_\mathcal{G}^2(0, T;
\mathbb{R}^m) \times L_\mathcal{G}^2(0, T; \mathbb{R}^{m\times d})
\times L_\mathcal{G}^{2, \tau}(0, T; \mathbb{R}^{m\times k})$, we
set
$$(Y, Z, \zeta)=I[(y, z, \varsigma)],\ \ (\bar{Y}, \bar{Z}, \bar{\zeta})=I[(\bar{y}, \bar{z}, \bar{\varsigma})],$$
and denote their differences by $(\hat{y}, \hat{z},
\hat{\varsigma})=(y-\bar{y}, z-\bar{z}, \varsigma-\bar{\varsigma})$,
$(\hat{Y},\hat{Z}, \hat{\zeta})=(Y-\bar{Y},Z-\bar{Z},
\zeta-\bar{\zeta})$. By the basic estimate (\ref{estimat2}), we have
$$E\int_0^T(\frac{\beta}{2}|\hat{Y}_s|^2+|\hat{Z}_s|^2+\|\hat{\zeta_s}\|_\tau^2)e^{\beta s}ds\leq \frac{2}{\beta}E\int_0^T|g(s, y_s, z_s, \varsigma_s)-g(s, \bar{y}, \bar{z}_s, \bar{\varsigma}_s)|^2e^{\beta s}ds.$$
Since $g$ satisfies the Lipschitz condition, we have
$$E\int_0^T(\frac{\beta}{2}|\hat{Y}_s|^2+|\hat{Z}_s|^2+\|\hat{\zeta_s}\|_\tau^2)e^{\beta
s}ds\leq
\frac{6C^2}{\beta}E\int_0^T(|\hat{y}_s|^2+|\hat{z}_s|^2+\|\hat{\varsigma_s}\|_\tau^2)e^{\beta
s}ds.$$ Let $\beta=12(C^2+1)$, then we get
$$E\int_0^T(|\hat{Y}_s|^2+|\hat{Z}_s|^2+\|\hat{\zeta_s}\|_\tau^2)e^{\beta
s}ds\leq
\frac{1}{2}E\int_0^T(|\hat{y}_s|^2+|\hat{z}_s|^2+\|\hat{\varsigma_s}\|_\tau^2)e^{\beta
s}ds,$$ or
$$\|(\hat{Y},\hat{Z}, \hat{\zeta})\|_\beta\leq \frac{1}{\sqrt{2}}\|(\hat{y}, \hat{z},
\hat{\varsigma})\|_\beta.$$ Thus $I$ is a strict contraction mapping
of $L_\mathcal{G}^2(0, T; \mathbb{R}^m) \times L_\mathcal{G}^2(0, T;
\mathbb{R}^{m\times d}) \times L_\mathcal{G}^{2, \tau}(0, T;
\mathbb{R}^{m\times k})$. It follows by the fixed point theorem that
BSDE (\ref{3.1}) has a unique solution. From (a) and (b), obviously
$g(\cdot, Y_\cdot, Z_\cdot, \zeta_\cdot)\in L_\mathcal{G}^2(0, T;
\mathbb{R}^m)$. Thus by Lemma \ref{lemma:estimate}, we have $Y\in
S_\mathcal{G}^2(0, T; \mathbb{R}^m)$.
\end{proof}

\begin{remark}\label{remark:generator g}
In the above theorem, from the conditions that the generator $g$
satisfies, we know that here $g$ is independent of the last element
$\varsigma$ after the default occurs, i.e., $g(t, y, z,
\varsigma)\equiv g(t, y, z)$ on $t\in [\tau \wedge T, T].$ Its
financial explanation is that after the default occurs the influence
factor on the contingent claim is apart from the defaultable risky
part absolutely.
\end{remark}

\begin{remark}\label{remark:unique}
The solution of (\ref{3.1}) is unique, that is to say, if both $(Y,
Z, \zeta)$ and $(\bar{Y}, \bar{Z}, \bar{\zeta})\in
S_\mathcal{G}^2(0, T; \mathbb{R}^m) \times L_\mathcal{G}^2(0, T;
\mathbb{R}^{m\times d}) \times L_\mathcal{G}^{2, \tau}(0, T;
\mathbb{R}^{m\times k})$ satisfy (\ref{3.1}), then
$$E\int_0^T|Y_t-\bar{Y}_t|^2dt=0,
E\int_0^T|Z_t-\bar{Z}_t|^2dt=0, E\int_0^T|\zeta_t-\bar{\zeta}_t|^2
1_{\{\tau>t\}}\gamma_tdt=0.$$
\end{remark}

\begin{remark}\label{remark:z^2}
The uniqueness of $\{\zeta_t; t \in [0, T] \}$ can be explained in
this way: $\{\zeta_t; t \in [0, T] \}$ is unique, that is to say,
$\{\zeta_t; t \in [0, T] \}$ can only be uniquely determined on the
random interval $[0, \tau \wedge T]\cap \{t:\gamma_t\neq 0\}$, i.e,
its effective definition domain is just the set $[0, \tau \wedge
T]\cap \{t:\gamma_t\neq 0\}$. On the interval $[\tau \wedge T,
T]\cup \{t:\gamma_t= 0\}$, $\zeta_.$ can be arbitrary adapted random
process. In fact, this is a direct conclusion of the truth that
$dM_t\equiv 0$ on $[\tau \wedge T, T]\cup \{t:\gamma_t= 0\}$, indeed
$M_t\equiv 1-\int_0^{\tau \wedge T} \gamma_s ds$ on $[\tau \wedge T,
T]$.
\end{remark}

\subsection{Comparison theorem for 1-dimensional BSDEs with random default time}\label{sec:comparison}

Consider the following two 1-dimensional BSDEs with random default
time:
\begin{equation}\label{4.1}
Y_t=\xi(H_T)+\int_t^T g(s, Y_s, Z_s, \zeta_s)ds-\int_t^T
Z_sdB_s-\int_t^T\zeta_sdM_s,
\end{equation}
\begin{equation}\label{4.2}
\overline{Y}_t=\overline{\xi}(H_T)+\int_t^T \overline{g}_s
ds-\int_t^T \overline{Z}_sdB_s-\int_t^T\overline{\zeta}_sdM_s.
\end{equation}

For the generator function $g$, we introduce one more assumption:

(c)for each $(t, y, z)\in [0, T]\times\mathbb{R}\times\mathbb{R}^d$,
$(\varsigma, \bar{\varsigma})\in \mathbb{R}^k\times\mathbb{R}^k$,
$(\varsigma^i-\bar{\varsigma}^i)1_{\{\tau_i>t\}}\gamma_t^i \neq 0$,
the following holds:
$$\frac{g(t, y, z, \tilde{\varsigma}^{i-1})-g(t, y, z, \tilde{\varsigma}^{i})}{(\varsigma^i-\bar{\varsigma}^i)1_{\{\tau_i>t\}}\gamma_t^i} > -1,$$
where $\tilde{\varsigma}^i=(\bar{\varsigma}^1, \bar{\varsigma}^2,
\cdots, \bar{\varsigma}^i, \varsigma^{i+1}, \varsigma^{i+2}, \cdots,
\varsigma^k)$ and $\varsigma^i$ is the $i$-th component of
$\varsigma$.

\begin{theorem}\label{thm:comparison}
Suppose $\xi, \overline{\xi}$ satisfy the same assumptions as in
Theorem \ref{thm:1-main}, $g$ satisfies (a)-(c), $\overline{g}_s \in
L_\mathcal{G}^2(0, T; \mathbb{R})$. Let $(Y, Z, \zeta)$,
$(\overline{Y}, \overline{Z}, \overline{\zeta})$ be the unique
solutions of (\ref{4.1}), (\ref{4.2}) respectively. If
$$\xi \geq \overline{\xi},\ \ g(t, \overline{Y}_t, \overline{Z}_t, \overline{\zeta}_t) \geq \overline{g}_t,\ \ a.e., a.s.,$$
then
$$Y_t \geq \overline{Y}_t,\ \ \ \ a.e.,a.s..$$
Besides, the following holds true (the strict comparison theorem):
$$Y_0=\overline{Y}_0\ \Leftrightarrow \ \xi=\overline{\xi},\ \ g(t, \overline{Y}_t, \overline{Z}_t, \overline{\zeta}_t)\equiv\overline{g}_t.$$
\end{theorem}

\begin{proof}
Let $\hat{\xi}=\xi-\overline{\xi}$, $\hat{Y}_s=Y_s-\overline{Y}_s$,
$\hat{Z}_s=Z_s-\overline{Z}_s$,
$\hat{\zeta}_s=\zeta_s-\overline{\zeta}_s$, $\hat{g}_s=g(s,\
\overline{Y}_s,\ \overline{Z}_s,\
\overline{\zeta}_s)-\overline{g}_s$, then we have
\begin{equation*}
\left\{
\begin{tabular}{rll}
$-d \hat{Y}_s$ &=& $(a_s \hat{Y}_s+b_s \hat{Z}_s+c_s
\hat{\zeta}_s1_{\{\tau > s\}}\gamma_s+\hat{g}_s)ds-\hat{Z}_s dB_s-
\hat{\zeta}_s dM_s,$\\\\
$\hat{Y}_T$ &=& $\hat{\xi},$
\end{tabular}
\right.
\end{equation*}
where
\begin{equation*}
\begin{tabular}{rll}
& $a_s:=\left\{
\begin{tabular}{ll}
$\frac{g(s,\ Y_s,\ Z_s,\ \zeta_s)-g(s,\ \overline{Y}_s,\ Z_s,\
\zeta_s)}{Y_s-\overline{Y}_s},$ & \ \ \ if $Y_s\neq \overline{Y}_s,$
\\
$0,$ & \ \ \ if $Y_s = \overline{Y}_s,$
\end{tabular}
\right. $
\\
& $b_s:=\left\{
\begin{tabular}{ll}
$\frac{g(s,\ \overline{Y}_s,\ Z_s,\ \zeta_s)-g(s,\ \overline{Y}_s,\
\overline{Z}_s,\ \zeta_s)}{Z_s-\overline{Z}_s},$ & \ \ \ if $Z_s\neq
\overline{Z}_s,$ \\
$0,$ & \ \ \ if $Z_s= \overline{Z}_s,$
\end{tabular}
\right. $
\\
& $c_s^i:=\left\{
\begin{tabular}{ll}
$\frac{g(s,\ \overline{Y}_s,\ \overline{Z}_s,\
\widetilde{\zeta}_s^{i-1})-g(s,\ \overline{Y}_s,\ \overline{Z}_s,\
\widetilde{\zeta}^{i}_s)}{(\zeta_s^i-\overline{\zeta}_s^i)1_{\{\tau_i
> s\}}\gamma_s^i},$ & \ \ \ if
$(\zeta_s^i-\overline{\zeta}_s^i)1_{\{\tau_i
>
s\}}\gamma_s^i\neq 0,$ \\
$0,$ & \ \ \ if $(\zeta_s^i-\overline{\zeta}_s^i)1_{\{\tau_i >
s\}}\gamma_s^i=0.$
\end{tabular}
\right. $
\end{tabular}
\end{equation*}
Since $g$ satisfies (b) and (c), thus $|a_s| \leq C$, $|b_s|\leq C$
and $c_s^i > -1$. Set
$$Q_s:=\exp[\int_0^s \ln(1+c_u)dH_u-\int_0^s c_u1_{\{\tau > u\}}\gamma_u du+\int_0^s b_u dB_u-\frac{1}{2}\int_0^s b_u^2du+\int_0^s a_u du],$$
i.e,
\begin{equation*}
\left\{
\begin{tabular}{rll}
$dQ_s$ &=& $Q_{s-}(a_s ds+b_s dB_s+c_s dM_s),$\\\\
$Q_0$ &=& $1.$
\end{tabular}
\right.
\end{equation*}
Applying It\^{o}'s formula on $Q_s\hat{Y}_s$, we have
\begin{eqnarray*}
dQ_s\hat{Y}_s &=& -Q_s (c_s\hat{\zeta}_s1_{\{\tau >
s\}}\gamma_s+\hat{g}_s)ds+Q_s(\hat{Z}_s+b_s\hat{Y}_s)dB_s\\
&& +Q_{s-}(\hat{\zeta}_s+c_s\hat{Y}_{s-})dM_s+Q_{s-}c_s\hat{\zeta}_sdH_s\\
&=& -Q_s
\hat{g}_sds+Q_s(\hat{Z}_s+b_s\hat{Y}_s)dB_s+Q_{s-}(\hat{\zeta}_s+c_s\hat{Y}_{s-}+c_s\hat{\zeta}_s)dM_s.
\end{eqnarray*}
Integrate from $t$ to $T$ and take conditional expectation w.r.t
$\mathcal{G}_t$ on both sides:
$$Q_t\hat{Y}_t=E^{\mathcal{G}_t}[Q_T\hat{Y}_T+\int_t^T Q_s\hat{g}_sds]\geq 0,\ \ a.e., a.s..$$
Then $\hat{Y}_t \geq 0$ immediately follows.
\end{proof}

\begin{remark}\label{remark:c_s}
In the above, the definition of $c_s$ is proper. For simplicity, we
only discuss the case $k=1$. Indeed, for the case when
$(\zeta_s-\overline{\zeta}_s)1_{\{\tau
> s\}}\gamma_s=0$, we should have the following equality:
\begin{eqnarray*}
& & a_s \hat{Y}_s+b_s \hat{Z}_s+c_s \hat{\zeta}_s1_{\{\tau >
s\}}\gamma_s+\hat{g}_s\\
& & = g(s, Y_s, Z_s, \zeta_s)-g(s, \overline{Y}_s, \overline{Z}_s,
\zeta_s)+g(s, \overline{Y}_s, \overline{Z}_s,
\overline{\zeta}_s)-\overline{g}_s\\
&& = g(s, Y_s, Z_s, \zeta_s)-\overline{g}_s,
\end{eqnarray*}
which will hold if
$$(\zeta_s-\overline{\zeta}_s)1_{\{\tau
> s\}}\gamma_s=0\Rightarrow \zeta_s-\overline{\zeta}_s=0.$$
For this we can refer to Remark \ref{remark:generator g} and
\ref{remark:z^2}, more detailedly, if $1_{\{\tau
> s\}}=0$ then $g(s,\ \overline{Y}_s,\ \overline{Z}_s,\
\overline{\zeta}_s)\equiv g(s,\ \overline{Y}_s,\ \overline{Z}_s)$,
and if $\gamma_s=0$ then $\overline{\zeta}$ can be arbitrary and we
can choose $\overline{\zeta}=\zeta$.
\end{remark}

\begin{remark}\label{remark:g}
Condition (c) for the generator $g$ is significant for the
comparison theorem. In the following we give an example which
indicates that the strict comparison theorem will not hold if $g$
does not satisfy (c).
\end{remark}

\begin{example}\label{counterexample}
Suppose that k=1, $\xi=H_T$, $\overline{\xi}=0$, $g(t, y, z,
\varsigma)=1_{\{\tau > t\}}\sqrt{\gamma_t}-1_{\{\tau >
t\}}\sqrt{\gamma_t}(\sqrt{\gamma_t}+1)\varsigma$,
$\overline{g}_t=0$.
Clearly $g$ does not satisfy (c). Consider the following two BSDEs:
\begin{equation}\label{4.3}
Y_t=H_T+\int_t^T (1_{\{\tau > s\}}\sqrt{\gamma_s}-1_{\{\tau >
s\}}\sqrt{\gamma_s}(\sqrt{\gamma_s}+1)\zeta_s) ds-\int_t^T
Z_sdB_s-\int_t^T\zeta_sdM_s,
\end{equation}
\begin{equation}\label{4.4}
\overline{Y}_t=0-\int_t^T
\overline{Z}_sdB_s-\int_t^T\overline{\zeta}_sdM_s.
\end{equation}
It is easy to check that $(H_t, 0, 1)$, $(0, 0, 0)$ are the unique
solutions of (\ref{4.3}), (\ref{4.4}) respectively.
\\
Then we have
$$\xi= H_T \geq \overline{\xi}= 0, g(t, \overline{Y}_t, \overline{Z}_t,
\overline{\zeta}_t)=g(t, 0, 0, 0)=1_{\{\tau > t\}}\sqrt{\gamma_t}
\geq 0 =\overline{g}_t,$$ while in the meantime we get
$$Y_0=\overline{Y}_0=0, but\ P(\xi > \overline{\xi}) > 0\ and\ (L\times P)(g(t, \overline{Y}_t, \overline{Z}_t,
\overline{\zeta}_t) > \overline{g}_t) > 0, $$ where $L$ denotes
Lebesgue measure.
\end{example}

The comparison theorem, which allows us to compare the solutions of
two BSDEs with random default time, can ensure the attainability of
the upper price of a contingent claim in the evaluation/hedging
problem. The main idea can be seen in the next section.

\section{Application in zero-sum stochastic differential game
problem}\label{sec:game}

We are now going to study the link between the zero-sum stochastic
differential games in the defaultable setting and the BSDEs with
random default time studied in the previous section. First let us
describe the framework of the zero-sum game we consider.

Assume here that $m=d=k$ and $\gamma_t\geq 0$ is bounded. Let $x_0
\in \mathbb{R}^m$ and let $X_t$ be the solution of the following
stochastic differential equation:
$$X_t= x_0+ \int_0^t \sigma(s, X_{s-}) dB_s+ \int_0^t \kappa(s, X_{s-}) dM_s,$$
where the mapping $\sigma: [0, T]\times \mathbb{R}^m \rightarrow
\mathbb{R}^{m\times m}$ and $\kappa: [0, T]\times \mathbb{R}^m
\rightarrow \mathbb{R}^{m\times m}$ satisfy the following
assumptions:

(i) for $1 \leq i, j \leq m$, $\sigma_{ij}$ and $\kappa_{ij}$ are
progressively measurable;

(ii) for any $(t, x)\in [0, T]\times \mathbb{R}^m$, there exists a
constant $C_1 \geq 0$ such that $$|\sigma(t, x)|+|\kappa(t, x)|\leq
C_1(1+|x|);$$

(iii) for any $(t, x), (t, y) \in [0, T]\times \mathbb{R}^m$, there
exists a constant $C_2\geq 0$ such that $$|\sigma(t, x)-\sigma(t,
y)|+|\kappa(t, x)-\kappa(t, y)|\leq C_2|x-y|;$$

(iv) $\sigma(t, x)$, $\kappa(t, x)$ are invertible and
$\sigma^{-1}(t, x)$, $\kappa^{-1}(t, x)$ are bounded.

Then the process $\{X_t; t\in [0, T]\}$ exists and is unique.

Let $U$ (resp. $V$) be a compact metric space and $\mathcal{U}$
(resp. $\mathcal{V}$) be the space of all progressively measurable
processes $u=(u_t)_{t\in [0,T]}$ (resp. $v=(v_t)_{t\in [0, T]}$)
with values in  $U$ (resp. $V$).

Let the drift function $b$ map $[0, T]\times \mathbb{R}^m \times U
\times V$ into $\mathbb{R}^m$. Furthermore, $b$ is supposed to
satisfy:

(i) $b$ is $\mathcal{B}([0, T]\times \mathbb{R}^m \times U \times
V)$-measurable;

(ii) $b(t, x, u, v)$ is bounded for any $(t, x, u, v)$;

(iii) for any $(t, x) \in [0, T]\times \mathbb{R}^m$, $b(t, x,
\cdot, \cdot)$ is continuous on $U\times V$.

Now for each $u \in \mathcal{U}$, $v \in \mathcal{V}$, let $L^{u,
v}$ be the positive local martingale solution of:
\begin{equation*}
\left\{
\begin{tabular}{rll}
$dL_t^{u, v}$ &=& $L_{t-}^{u, v}(\sigma^{-1}(t, X_{t-})b(t, X_{t-},
u_t,
v_t)dB_t+\kappa^{-1}(t, X_{t-})c(t, X_{t-}, u_t, v_t)dM_t),$\\\\
$L_0^{u, v}$ &=& $1,$
\end{tabular}\right.
\end{equation*} where
for any $(t, x, u, v)$, $i=1, 2, \cdots, m$, the $i$-th component of
$\kappa^{-1}(t, x)c(t, x, u, v)$ is larger than $-1$, i.e.,
$(\kappa^{-1}(t, x)c(t, x, u, v))^i
> -1$.

According to the Girsanov Theorem (see Appendix), $P^{u, v}$ defined
by $\frac{dP^{u, v}}{dP}|\mathcal{G}_T=L_T^{u, v}$ is a probability
measure equivalent to $P$. Moreover, under $P^{u, v}$, the process
$B_t^{u, v}=B_t-\int_0^t \sigma^{-1}(s, X_{s-})b(s, X_{s-}, u_s,
v_s) ds$ is a Brownian motion, the processes $M_t^{i, u,
v}=M_t^i-\int_0^t (\kappa^{-1}(s, X_{s-})c(s, X_{s-}, u_s,
v_s))^i1_{\{\tau_i>s\}}\gamma_s^ids$ $(i=1, 2, \cdots, m)$ are
$\mathbb{G}$-martingales orthogonal to each other and orthogonal to
$B_t^{u, v}$ and $(X_t)_{0\leq t \leq T}$ satisfies
\begin{equation*}
\left\{
\begin{tabular}{rll}
$dX_t$ &=& $(b(t, X_{t-}, u_t, v_t)+c(t, X_{t-}, u_t, v_t)1_{\{\tau>
t\}}\gamma_t) dt$
\\\\
& & $+\sigma(t, X_{t-})dB_t^{u,v}+\kappa(t, X_{t-})dM_t^{u, v},$\\\\
$X_0$ &=& $x_0.$
\end{tabular}\right.
\end{equation*}
It means that $(X_t)_{0\leq t \leq T}$ is a weak solution for the
above stochastic differential equation and it stands for an
evolution of a controlled system.

It is well-known that in zero-sum game problems, there are two
players $J_1$ and $J_2$. We suppose that $J_1$ (resp. $J_2$) chooses
a control $u(t, x)\in U$ (resp. $v(t, x)\in V$). Now we introduce
two functions $f: [0, T]\times \mathbb{R}^m \times U \times V
\rightarrow \mathbb{R}_+$, satisfying the same assumptions as $b$,
and $h: \{0, 1\}\times \mathbb{R}^m \rightarrow \mathbb{R}_+$ which
is measurable, bounded. Let $E^{u, v}$ denote the expectation w.r.t
$P^{u, v}$. Then the cost functional corresponding to $u \in
\mathcal{U}$ and $v \in \mathcal{V}$ is given by
$$J^{u, v}= E_{u, v} [\int_0^T f(s, X_s, u_s, v_s)ds+ h(H_T, X_T)],$$
which is a cost (resp. reward) for $J_1$ (resp. $J_2$).

The object of $J_1$ (resp. $J_2$) is  to minimize (resp. maximize)
the cost functional. In this zero-sum game problem, we aim at
showing the existence of a saddle point, more precisely, a pair
$(\tilde{u}^*, \tilde{v}^*)$ such that $J(\tilde{u}^*, v) \leq
J(\tilde{u}^*, \tilde{v}^*) \leq J(u, \tilde{v}^*)$ for each $(u,
v)\in \mathcal{U}\times \mathcal{V}$.

Thus let us define the Hamilton function associated with this game
problem as following: $\forall (t, x, z, \varsigma, u, v) \in [0,
T]\times \mathbb{R}^m \times \mathbb{R}^{m} \times \mathbb{R}^{m}
\times U \times V, $
\begin{eqnarray*}
H(t, x, z, \varsigma, u, v)&:=&z \sigma^{-1}(t, x)b(t, x, u, v)+
\varsigma \kappa^{-1}(t, x)c(t, x, u, v) 1_{\{\tau> t\}}\gamma_t\\
& & + f(t, x, u, v).
\end{eqnarray*}
Here we should pay special attention to the difference between the
notations of the Hamilton function $H(t, \cdot, \cdot, \cdot, \cdot,
\cdot)$ and the default process $H_t$.

Next assume that Isaacs' condition, which plays an important role in
zero-sum stochastic differential game problems, is fulfilled, i.e,
for any $(t, x, z, \varsigma) \in [0, T]\times \mathbb{R}^m \times
\mathbb{R}^{m} \times \mathbb{R}^{m}$,
$$\inf_{u\in U}\sup_{v\in V}H(t, x, z, \varsigma, u, v)=\sup_{v\in V}\inf_{u\in U}H(t, x, z, \varsigma, u, v).$$

Under the above Isaacs' condition, through the assumptions above and
Benes's selection theorem (see e.g. \cite{B}), the following holds
true (see e.g. \cite{KH}).

\begin{proposition}\label{prop:isaacs}
There exist two measurable functions $u^*(t, x, z, \varsigma)$,
$v^*(t, x, z, \varsigma)$ mapping from $[0, T]\times \mathbb{R}^m
\times \mathbb{R}^{m} \times \mathbb{R}^{m}$ into $U$, $V$
respectively such that:
\item{(i)} the pair $(u^*, v^*)(t, x, z, \varsigma)$ is a saddle
point for the function $H$, i.e,
$$H(t, x, z, \varsigma, u^*(t, x, z,
\varsigma), v^*(t, x, z, \varsigma)) \leq H(t, x, z, \varsigma, u,
v^*(t, x, z, \varsigma)),\ \ \forall\ u\in U,$$
$$H(t, x, z, \varsigma,
u^*(t, x, z, \varsigma), v^*(t, x, z, \varsigma)) \geq H(t, x, z,
\varsigma, u^*(t, x, z, \varsigma), v),\ \ \forall\ v\in V;$$
\item{(ii)} the function $(z, \varsigma)\rightarrow H(t, x, z, \varsigma,
u^*(t, x, z, \varsigma), v^*(t, x, z, \varsigma))$ satisfies (b) and
(c), uniformly in $(t, x)$.
\end{proposition}

Now we introduce two notations just for simplicity:
\begin{eqnarray*}
& &H(t, z, \varsigma) := H(t, X_{t-}, z, \varsigma, u_t, v_t),\\
& &H^*(t, z, \varsigma) := H(t, X_{t-}, z, \varsigma, u^*(t, X_{t-},
z, \varsigma), v^*(t, X_{t-}, z, \varsigma)).
\end{eqnarray*}

Suppose that $J_1$ (resp. $J_2$) has chosen $u\in \mathcal{U}$
(resp. $v\in \mathcal{V}$). The conditional expected remaining cost
from time $t\in [0, T]$ is
$$J_t^{u, v}=E_{u, v}^{\mathcal{G}_t}[\int_t^T f(s, X_s, u_s, v_s)ds+ h(H_T, X_T)].$$
It is obvious that $J_0^{u, v}=J^{u, v}.$ The following theorem
tells us that the conditional costs can be characterized as
solutions of BSDEs with random default time.

\begin{theorem}\label{thm:game and bsde}
The BSDE with random default time
\begin{equation}\label{gamebsde}
Y_t=h(H_T, X_T)+\int_t^T H(s, Z_s, \zeta_s)ds- \int_t^T
Z_sdB_s-\int_t^T \zeta_sdM_s
\end{equation}
has a unique solution $(Y, Z, \zeta)\in S_\mathcal{G}^2(0, T;
\mathbb{R}) \times L_\mathcal{G}^2(0, T; \mathbb{R}^{m}) \times
L_\mathcal{G}^{2, \tau}(0, T; \mathbb{R}^{m})$ which satisfies
$Y_t=J_t^{u, v}$.
\end{theorem}

\begin{proof}
Notice that
\begin{eqnarray*}
H(t, z, \varsigma)&=&H(t, X_{t-}, z, \varsigma, u_t,
v_t)\\
&=&z \sigma^{-1}(t, X_{t-})b(t, X_{t-}, u_t, v_t)+
\varsigma\kappa^{-1}(t, X_{t-})c(t, X_{t-}, u_t, v_t) 1_{\{\tau>
t\}}\gamma_t\\
& &+ f(t, X_{t-}, u_t, v_t).
\end{eqnarray*}
Then (\ref{gamebsde}) can be transformed to the following:
\begin{eqnarray*}
Y_t&=&h(H_T, X_T)+\int_t^T f(s, X_s, u_s, v_s)ds\\
&& -\int_t^T Z_s d(B_s-\int_0^s \sigma^{-1}(r, X_r)b(r, X_r, u_r, v_r)dr)\\
&& -\int_t^T \zeta_s d(M_s-\int_0^s \kappa^{-1}(r, X_{r-})c(r,
X_{r-}, u_r, v_r)1_{\{\tau> r\}}\gamma_rdr).
\end{eqnarray*}
According to the Girsanov Theorem, we can easily obtain
$$Y_t=E_{u, v}^{\mathcal{G}_t}[\int_t^T f(s, X_s, u_s, v_s)ds+ h(H_T, X_T)],$$
i.e., $Y_t=J_t^{u, v}.$
\end{proof}
\\
\par Next is the main result of this part.

\begin{theorem}\label{thm:game and bsde 2}
The BSDE with random default time
\begin{equation}\label{gamebsde2}
Y_t=h(H_T, X_T)+\int_t^T H^*(s, Z_s, \zeta_s)ds- \int_t^T
Z_sdB_s-\int_t^T \zeta_sdM_s.
\end{equation}
has a unique solution $(Y, Z, \zeta)\in S_\mathcal{G}^2(0, T;
\mathbb{R}) \times L_\mathcal{G}^2(0, T; \mathbb{R}^{m}) \times
L_\mathcal{G}^{2, \tau}(0, T; \mathbb{R}^{m})$ which satisfies
$$Y_t=J_t^{\tilde{u}^*, \tilde{v}^*},$$
where $\tilde{u}^*(t, X_{t-})=u^*(t, X_{t-}, Z_t, \zeta_t),$
$\tilde{v}^*(t, X_{t-})=v^*(t, X_{t-}, Z_t, \zeta_t).$ Moreover, the
pair $(\tilde{u}^*, \tilde{v}^*)$ is a saddle point for the game.
\end{theorem}

\begin{proof}
By Proposition \ref{prop:isaacs} (ii) and Theorem \ref{thm:1-main},
we can easily get the existence and uniqueness of the solution $(Y,
Z, \zeta)$ to (\ref{gamebsde2}). Similarly as in Theorem
\ref{thm:game and bsde}, we have $Y_t=J_t^{\tilde{u}^*,
\tilde{v}^*}$.

Next we prove that $(\tilde{u}^*, \tilde{v}^*)$ is a saddle point
for the game. The main tool we use is the comparison theorem for
BSDEs with random default time.

First let us consider the following equation:
\begin{equation}\label{gamebsde3}
Y_t=h(H_T, X_T)+\int_t^T H(s, X_{s-}, Z_s, \zeta_s, u_s,
\tilde{v}^*)ds-\int_t^T Z_sdB_s-\int_t^T \zeta_sdM_s.
\end{equation}

By Proposition \ref{prop:isaacs} (ii), (\ref{gamebsde3}) has a
unique solution $(Y^u, Z^u, \zeta^u)$ with $Y_t^u=J_t^{u,
\tilde{v}^*}$. By Proposition \ref{prop:isaacs} (i), we have $H(t,
X_{t-}, Z_t, \zeta_t, u_t, \tilde{v}^*) \geq H^*(t, Z_t, \zeta_t)$
for each $u\in \mathcal{U}$. It then follows from the comparison
theorem that for each $t\in [0, T]$, $J_t^{u, \tilde{v}^*}\geq
J_t^{\tilde{u}^*, \tilde{v}^*}$a.s., for each $u\in \mathcal{U}$.
Then $J^{u, \tilde{v}^*}\geq J^{\tilde{u}^*, \tilde{v}^*}$ for each
$u\in \mathcal{U}$.

In a symmetric way, considering
\begin{equation*}
Y_t=h(H_T, X_T)+\int_t^T H(s, X_{s-}, Z_s, \zeta_s, \tilde{u}^*,
v_s)ds- \int_t^T Z_sdB_s- \int_t^T \zeta_sdM_s,
\end{equation*}
we obtain $J_t^{\tilde{u}^*, v}\leq J_t^{\tilde{u}^*, \tilde{v}^*}$
and $J^{\tilde{u}^*, v}\leq J^{\tilde{u}^*, \tilde{v}^*}$ for each
$v\in \mathcal{V}$.

Therefore $J^{\tilde{u}^*, v}\leq J^{\tilde{u}^*, \tilde{v}^*}\leq
J^{u, \tilde{v}^*}$ for each $u\in \mathcal{U}$ and $v\in
\mathcal{V}$, i.e, the pair $(\tilde{u}^*, \tilde{v}^*)$ is a saddle
point for the game.
\end{proof}

\begin{remark}\label{gameutility}
From another point of view, the game problem is just a utility
maximization problem under model uncertainty, once $\mathcal{U}$ is
regarded as the set that leads to model uncertainty and
$\mathcal{V}$ the set of admissible trading strategies for the
investor.
\end{remark}

\begin{remark}\label{gamecontrol}
The results can be applied to the control problem (as well as the
utility maximization problem)
of the existence of an optimal strategy where the diffusions are
bounded. For this, we can choose $f$, $b$ and $c$ independent of
$u$. Then we get that $\tilde{v}^*(t, X_{t-})=\tilde{v}^*(t, X_{t-},
Z_t, \zeta_t)$ is an optimal strategy for the optimal stochastic
control problem, where $\tilde{v}^*(t, x, z)$ maximizes
$$H(t, x, z, \varsigma, v)=z \sigma^{-1}(t, x)b(t, x, v)+ \varsigma
\kappa^{-1}(t, x)c(t, X_{t-}, v)1_{\{\tau> t\}}\gamma_t+ f(t, x,
v),$$ and where $(Z_t, \zeta_t)$ is such that $(Y_t, Z_t, \zeta_t)$
is the unique solution of
$$Y_t=h(H_T, X_T)+\int_t^T H(s, X_{s-}, Z_s, \tilde{v}_s^*)ds-\int_t^T
Z_sdB_s-\int_t^T \zeta_s dM_s.$$ It should be noted that in the
utility maximization problem, we always choose $f=0$ and call $h$
the utility function. Besides, it is obvious that the value of the
utility maximization problem just equals to $Y_0$.
\end{remark}

\section*{Acknowledgements}

The first author thanks Prof. Monique Jeanblanc for very helpful
discussions during her visit in Shandong University. The second
author is grateful to Shuai Jing for his careful reading and useful
suggestions.



\appendix\section*
{Appendix: Some basic results}

\renewcommand{\theequation}{A.\arabic{theorem}}
\setcounter{section}{1}\setcounter{theorem}{0}
\renewcommand{\theequation}{A.\arabic{remark}}
\setcounter{section}{1}\setcounter{remark}{0}
\renewcommand{\theequation}{A.\arabic{equation}}
\setcounter{section}{1}\setcounter{equation}{0}

Let us recall some basic and essential results for this paper. Note
that all are in a defaultable setting.

\begin{theorem}\label{thm:general ito formula} (It\^{o}'s
formula). Let $X_t$ be an m-dimensional It\^{o} jump-diffusion
process given by
$$dX_t=b_tdt+\sigma_t dB_t+ \kappa_t dM_t,$$
where $B_t$ is a $d$-dimensional Brownian motion, $M_t$ is a
$k$-dimensional martingale (i.e., there are $k$ default times
$\tau_1$, $\tau_2$, $\cdots$, $\tau_k$),  $b_t$, $\sigma_t$ and
$\kappa_t$ are $\mathbb{G}$-adapted processes with corresponding
dimensions satisfying
$$E\int_0^T |b_t|dt< + \infty,\ E\int_0^T |\sigma_t|^2dt< + \infty,\ E\int_0^T |\kappa_t|^2 1_{\{\tau>t\}}\gamma_tdt < + \infty.$$ Let
$f(t, x)\in C^{1, 2}([0, T]\times \mathbb{R}^m; \mathbb{R})$.
\\
Then the process
$$Y_t:=f(t, X_t)$$
is again an It\^{o} jump-diffusion process, and it can be given by
\begin{equation}\label{multi- dim ito formula differential}
\begin{tabular}{rll}
$dY_t$ &=& $\frac{\partial f}{\partial t}(t, X_t)dt+\sum_{i=1}^m
\frac{\partial f}{\partial x_i}(t,
X_t)dX_t^i$\\\\
&& $+\frac{1}{2}\sum_{i, j=1}^m\sum_{k=1}^d \frac{\partial^2
f}{\partial x_i \partial x_j}(t, X_t)\sigma_{ik} \sigma_{jk}dt$
\\\\
&&  $+\sum_{j=1}^k [\Delta_{j} f(t,
X_{t-})-\sum_{i=1}^m\frac{\partial f}{\partial x_i}(t,
X_{t-})\kappa_{ij, t}]dH_t^j,$
\end{tabular}
\end{equation}
where
\begin{equation*}
\Delta_{j} f(t, X_{t-}) := f(t, X_{t-}^1+\kappa_{1j, t}, \cdots,
X_{t-}^i+\kappa_{ij, t}, \cdots, X_{t-}^m+\kappa_{mj, t})-f(t,
X_{t-}).
\end{equation*}
\end{theorem}

The main idea of the proof can be referred to \cite{Meyer},
\cite{YHW} or \cite{Protter}. Here we only give a sketch of proof
for reader's convenience.
\\
{\bf{Sketch of proof.}} For the sake of simplicity, we only give the
proof for the case when $m=d=k=1$. We know that the jump-diffusion
process $X$ jumps only at $\tau$ with the jump size $\kappa_\tau$,
thus on [0, $\tau \wedge T$) and ($\tau \wedge T$, T],
$$dX_t=dX_t^c=b_tdt+\sigma_t dB_t - \kappa_t 1_{\{\tau> t\}}\gamma_tdt.$$
Applying the It\^{o}'s formula in the Brownian case, we obtain
\begin{equation*}
Y_t-Y_0 = \int_0^t [(\frac{\partial f}{\partial s}(s,
X_s)+\frac{1}{2}\sigma_s^2\frac{\partial^2 f}{\partial x^2}(s,
X_s)]ds+\int_0^t \frac{\partial f}{\partial x}(s, X_s)dX_s^c,
\end{equation*}
since $dX_s=dX_s^c$ on $[0, t] \subset [0, \tau \wedge T);$
\begin{equation*}
Y_t-Y_r = \int_r^t [(\frac{\partial f}{\partial s}(s,
X_s)+\frac{1}{2}\sigma_s^2\frac{\partial^2 f}{\partial x^2}(s,
X_s)]ds+ \int_r^t \frac{\partial f}{\partial x}(s, X_s)dX_s^c,
\end{equation*}
since $dX_s=dX_s^c$ on $[r, t]\subset (\tau \wedge T, T]$.

If the default event occurs at the default time $\tau$ with jump
size $\kappa_\tau$, then the resulting change in $Y_t$ is given by
$f(\tau, X_\tau)-f(\tau, X_{\tau-})=f(\tau,
X_{\tau-}+\delta_\tau)-f(\tau, X_{\tau-})$.

Thus the total change in $Y_t$ can be written as the sum of these
two contributions:
\begin{equation*}
\begin{tabular}{rll}
$Y_t-Y_0$ &=& $\int_0^t [(\frac{\partial f}{\partial s}(s,
X_s)+\frac{1}{2}\sigma_s^2\frac{\partial^2 f}{\partial x^2}(s,
X_s))ds+\frac{\partial f}{\partial x}(s, X_s)dX_s^c]$ \\\\
& & $+1_{\{\tau \leq t\}}[f(\tau, X_\tau)-f(\tau, X_{\tau-})]$
\\\\
&=& $\int_0^t [\frac{\partial f}{\partial s}(s,
X_s)+(b_s-\kappa_s1_{\{\tau> s\}}\gamma_s)\frac{\partial f}{\partial
x}(s, X_s)+\frac{1}{2}\sigma_s^2 \frac{\partial^2 f}{\partial
x^2}(s, X_s)]ds$
\\\\
& & $+\int_0^t \frac{\partial f}{\partial x}(s, X_{s})\sigma_s
dB_s+\int_0^t [f(s, X_s)-f(s, X_{s-})]dH_s.$
\end{tabular}
\end{equation*}

Note that for any Borel measurable function $g$, we have $\int_0^t
g(s, X_{s-})ds=\int_0^t g(s, X_s)ds$ since $X_{s-}$ and $X_s$ differ
only for at most one value of $s$ (for each $\omega\in \Omega$). By
simple computation, we can get
\begin{equation*}
\begin{tabular}{rll}
$Y_t-Y_0$ &=& $\int_0^t \frac{\partial f}{\partial s}(s,
X_s)ds+\int_0^t \frac{\partial f}{\partial x}(s, X_s)dX_s+\int_0^t
\frac{1}{2}\sigma_s^2 \frac{\partial^2 f}{\partial
x^2}(s, X_s)ds$\\\\
& & $+\int_0^t [(f(s, X_s)-f(s, X_{s-}))-\kappa_s \frac{\partial
f}{\partial x}(s, X_{s-})]dH_s.\ \ \ \ \ \ \ \ \ \Box$
\end{tabular}
\end{equation*}

\begin{example}\label{eg:integration by parts}
As an application of It\^{o}'s formula, we compute $f(t,
X_t)=e^{\beta t}X_t^2.$ Obviously $f(t, x)=e^{\beta t} x^2 \in C^{1,
2}([0, T]\times \mathbb{R}; \mathbb{R})$, and we have
$$\frac{\partial f}{\partial t}(t, x)=\beta e^{\beta t} x^2,\ \
\frac{\partial f}{\partial x}(t, x)=2e^{\beta t} x,\ \
\frac{\partial^2 f}{\partial x^2}(t, x)=2e^{\beta t}.$$ Hence due to
It\^{o}'s formula (\ref{multi- dim ito formula differential}), we
obtain
\begin{equation}\label{eg: int. by parts 1}
\begin{tabular}{rll}
$d f(t, X_t)$ &=& $\beta e^{\beta t} X_t^2dt+2 e^{\beta t}
X_tdX_t+\frac{1}{2}\sigma_t^2 \cdot 2 e^{\beta t}dt$
\\\\
& & $+ [(e^{\beta t}X_t^2-e^{\beta t}X_{t-}^2)-\kappa_t \cdot 2
e^{\beta t} X_{t-}]dH_t$
\\\\
&=& $\beta e^{\beta t} X_t^2dt+2 e^{\beta t} X_tdX_t+\sigma_t^2
\cdot e^{\beta t}dt$
\\\\
& & $+ [(e^{\beta t}(X_{t-}+\kappa_t)^2-e^{\beta
t}X_{t-}^2)-\kappa_t \cdot 2 e^{\beta t} X_{t-}]dH_t$
\\\\
&=& $e^{\beta t} (\beta X_t^2+\sigma_t^2+\kappa_t^2 1_{\{\tau
>t\}} \gamma_t)dt+2 e^{\beta t} X_tdX_t+ e^{\beta t} \kappa_t^2
dM_t.$
\end{tabular}
\end{equation}
If we use the formula of integration by parts, first compute
$e^{\beta t} X_t$ as follows:
\begin{equation*}
\begin{tabular}{rll}
$de^{\beta t} X_t$ &=& $e^{\beta t}dX_t+X_t de^{\beta t}+d[e^{\beta
t},
X_t]$\\\\
&=& $e^{\beta t}(b_t+\beta X_t)dt+e^{\beta t}\sigma_tdB_t+e^{\beta
t}\kappa_tdM_t,$
\end{tabular}
\end{equation*}
secondly compute $f(t, X_t)=e^{\beta t}X_t^2=e^{\beta t}X_t \cdot
X_t$:
\begin{equation}\label{eg: int. by parts 2}
\begin{tabular}{rll}
$d f(t, X_t)$ &=& $d (e^{\beta t}X_t \cdot X_t)=e^{\beta t}X_t d
X_t+X_t d (e^{\beta t}X_t)+d[e^{\beta t}X_t, X_t]$
\\\\
&=& $e^{\beta t}X_t (b_tdt+\sigma_tdB_t+\kappa_tdM_t)$
\\\\
&& $+X_t e^{\beta t}[(b_t+ \beta
X_t)dt+\sigma_tdB_t+\kappa_tdM_t]+e^{\beta t}\sigma_t^2 dt+e^{\beta
t}\kappa_t^2dH_t$
\\\\
&=& $e^{\beta t} (\beta X_t^2+\sigma_t^2+\kappa_t^2 1_{\{\tau
>t\}} \gamma_t)dt+2 e^{\beta t} X_tdX_t+ e^{\beta t} \kappa_t^2
dM_t.$
\end{tabular}
\end{equation}

From the above, we can find that (\ref{eg: int. by parts 1}) and
(\ref{eg: int. by parts 2}) are of the same form. This indicates
that the formula of integration by parts is in fact a special case
of It\^{o}'s formula, which is well-known already in the Brownian
case.

\end{example}

The next theorem is Theorem 2.3 of Kusuoka \cite{K}:

\begin{theorem}\label{thm:maringale representation} (Martingale Representation Theorem).
Assume that both $(\bf{A})$ and $(\bf{H})$ hold. Then any
$\mathbb{G}-$square integrable martingale admits a representation as
the sum of a stochastic integral w.r.t the Brownian motion and
stochastic integrals w.r.t the martingales $\{M^i; i=1, 2, \ldots,
k\}$ associated with $\{\tau_i; i=1, 2, \ldots, k\}$ respectively.

More precisely, suppose $(N_t)_{0\leq t \leq T}$ is a
$\mathbb{G}$-square integrable martingale. Then there exist
$\mathbb{G}$-adapted processes $\mu_s: [0, T]\times \Omega
\rightarrow \mathbb{R}^d$ and $\nu_s^i: [0, T]\times \Omega
\rightarrow \mathbb{R}$ $(i=1, 2, \cdots, k)$ such that
\begin{equation}\label{equation:girsanov1}
E\int_0^T |\mu_s|^2ds < \infty,\ E\int_0^T |\nu_s^i|^2\gamma_s^ids <
\infty,\ i=1,\ 2,\ \cdots,\ k
\end{equation}
and
\begin{equation}\label{equation:girsanov2}
N_t=N_0+ \int_0^t \mu_sdB_s+\int_0^t \nu_sdM_s:=N_0+ \int_0^t
\mu_sdB_s+ \sum_{i=1}^k \int_0^t \nu_s^idM_s^i.
\end{equation}
\end{theorem}

\begin{remark}\label{kusuoka unique}
In fact, in Kusuoka's martingale representation theorem, the
processes $\mu(\cdot)$ and $\nu_i(\cdot)$ $(i=1, 2, \cdots, k)$ are
unique, that is to say, if processes $\tilde{\mu}: [0, T]\times
\Omega \rightarrow \mathbb{R}^d$ and $\tilde{\nu}^i: [0, T]\times
\Omega \rightarrow \mathbb{R}$ $(i=1, 2, \cdots, k)$ also make
(\ref{equation:girsanov1}) and (\ref{equation:girsanov2}) true, then
we undoubtedly have
$$E\int_0^T |\mu_s-\tilde{\mu}_s|^2ds=0,\ E\int_0^T |\nu_s^i-\tilde{\nu}_s^i|^21_{\{\tau_i>s\}}\gamma_s^ids=0,\ i=1, 2, \cdots, k.$$
\end{remark}

The Girsanov Theorem, stated below, can be referred to Kusuoka
\cite{K} (Proposition 3.1) or Bielecki et al. \cite{BJR5}
(Proposition 3.2.2).

\begin{theorem}\label{thm:girsanov} (Girsanov Theorem).
Let $Q$ be a probability measure on $(\Omega, \mathcal{G}_T)$
equivalent to $P$. If the Radon-Nikodym density $\eta.$ of $Q$ w.r.t
$P$ is given as follows:
\begin{equation*}
\eta_t:=\frac{dQ}{dP}|\mathcal{G}_t
=1+\int_0^t\eta_{s-}(\rho_sdB_s+\kappa_sdM_s)=1+\int_0^t\eta_{s-}(\rho_sdB_s+\sum_{i=1}^k
\kappa_s^idM_s^i),
\end{equation*}
where $\kappa^i > -1,\ i=1,\ 2,\ \cdots, k.$ Then the process
$$B_t^\ast=B_t-\int_0^t\rho_sds,\ \forall t\in[0, T]$$
follows a Brownian Motion w.r.t $\mathbb{G}$ under $Q$, and the
processes
$${M_t^{i, \ast}=M_t^i-\int_0^t\kappa_s^i1_{\{\tau^i>s\}}\gamma_s^ids,\ i= 1, 2, \ldots, k}$$
are $\mathbb{G}-$martingales orthogonal to each other and orthogonal
to $B^\ast$.
\end{theorem}

\end{document}